\documentclass[a4,11pt]{article}

\usepackage{microtype} 
\usepackage{tikz}
\usepackage{calc}
\usepackage[intlimits]{amsmath}
\usepackage{amsfonts}
\usepackage{amsthm}
\usepackage{amssymb}
\usepackage{mathtools}
\usepackage{subcaption}
\usepackage{caption}
\usepackage{hyperref}

\usetikzlibrary{shapes,positioning,chains,fit,calc,decorations.markings,decorations.pathreplacing}
\tikzstyle{ccyan}=[circle, draw, thick,fill=cyan!30, minimum size=12pt,inner sep=0pt]
\tikzstyle{cgrey}=[circle, draw, thick,fill=gray!30, minimum size=10pt,inner sep=0pt]
\tikzstyle{cgreys}=[circle, draw, thick,fill=gray!30, minimum size=12pt,inner sep=0pt]

\newcommand{\ket}[1]{| #1 \rangle}
\newcommand{\bra}[1]{\langle #1 |}

\newcommand{\overbar}[1]{\mkern 1.8mu\overline{\mkern-1.8mu#1\mkern-1.8mu}\mkern 1.8mu}

\newcommand{\comment}[1]{}

\makeatletter
\newtheorem*{rep@theorem}{\rep@title}
\newcommand{\newreptheorem}[2]{%
\newenvironment{rep#1}[1]{%
 \def\rep@title{#2 \ref{##1}}%
 \begin{rep@theorem}}%
 {\end{rep@theorem}}}
\makeatother

\newtheorem{theorem}{Theorem}
\newtheorem{corollary}{Corollary}
\newtheorem{lemma}{Lemma}
\newreptheorem{lemma}{Lemma}

\title{On the probability of finding marked connected components using quantum walks}

\author{Nikolajs Nahimovs$^1$, Raqueline A. M. Santos$^1$ and Kamil Khadiev$^{1,2}$}
\date{\small{$^1$Center for Quantum Computer Science, University of Latvia} \\ 
\small{Raina bulv. 19, Riga, LV-1586, Latvia}\\
\small{$^2$Kazan Federal University,} \\ 
\small{Kremlevskaya 18, Kazan, 420008, Russia}\\
\small{\texttt{nikolajs.nahimovs@lu.lv, rsantos@lu.lv, kamilhadi@gmail.com}}
}

\begin{document}

\maketitle


\begin{abstract}

\noindent
Finding a marked vertex in a graph can be a complicated task when using quantum walks.
Recent results show that for two or more adjacent marked vertices search by quantum walk with Grover's coin may have no speed-up over classical exhaustive search.
In this paper, we analyze the probability of finding a marked vertex for a set of connected components of marked vertices. We prove two upper bounds on the probability of finding a marked vertex and sketch further research directions. 

\end{abstract}

 
\section{Introduction}

Searching is an important problem in Computer Science. 
Using Grover's quantum algorithm~\cite{Grover:1996} one can solve the unstructured search problem quadratically faster than classically. A quadratic speed-up is also obtained when searching for a single marked vertex in some classes of graphs by using quantum walks~\cite{Szegedy:2004,Magniez:2012,Krovi:2016}. 
In case of multiple marked vertices the situation gets more tricky.
Krovi~{\it et al.}~\cite{Krovi:2016} gave a quantum walk based algorithm that achieves the quadratic speed-up for any reversible and ergodic Markov chain and showed that for multiple marked vertices it can search quadratically faster than a quantity called the ``extended hitting time'', which is equivalent to the hitting time for one marked vertex and lower-bounded by it.
Recently, Hoyer and Komeili~\cite{Hoyer:2017} described a quantum walk based algorithm for finding multiple marked vertices in the two-dimensional lattice. Their algorithm uses quadratically fewer steps than a random walk on the two-dimensional lattice, ignoring logarithmic factors.
On the other hand, for some quantum walk based search algorithms additional marked vertices can make the search easier or harder depending on the placement of marked vertices\cite{Nahimovs:2016}. 

In this paper, we consider search by coined discrete-time quantum walk~\cite{Aharonov:1993} on general graphs with multiple marked vertices. Suppose we have a graph $G = (V, E)$ with a set of vertices $V$ and a set of edges $E$. Let $n = |V|$ and $m = |E|$. 
The discrete-time quantum walk on $G$ has associated Hilbert space ${\cal H}^{2m}$ with the set of basis states $\{\ket{v,c}: v \in V, 0 \leq c < d_v \}$, where $d_v$ is the degree of vertex $v$. 
The evolution operator is the product of the coin operator followed by the shift operator, that is, $U= S\cdot C.$
The coin transformation $C$ is the direct sum of coin transformations for individual vertices, i.e. $C = C_{d_1}\bigoplus \cdots\bigoplus C_{d_n}$ with $C_{d_i}$ being the Grover diffusion transformation of dimension $d_i$.
The shift operator $S$ acts as $S\ket{v,c} = \ket{v',c'}$,
where $v$ and $v'$ are adjacent, $c$ and $c'$ represent the directions that points from $v$ to $v'$ and from $v'$ to $v$, respectively.  

Searching for a marked vertex is done using the unitary operator $U' = S \cdot C \cdot Q$,
where $Q$ is the query transformation, which flips the signs of the amplitudes at the marked vertices, that is,
\begin{equation}
Q = I - 2\sum_{w\in M}\sum_{c=0}^{d_w-1}\ket{w,c}\bra{w,c},
\end{equation}
with $M$ being the set of marked vertices. 

The initial state of the algorithm is the equal superposition over all vertex-direction pairs:
\begin{equation}\label{eq:psi0_gen}
\ket{\psi(0)} = \frac{1}{\sqrt{2m}} 
\sum_{v=0}^{n-1} \sum_{c = 0}^{d_v-1} \ket{v,c}.
\end{equation}
It can be easily verified that the initial state stays unchanged by the evolution operator $U$, regardless of the number of steps (the same holds for the search operator $U'$ is there are no marked vertices). 

In this model, Nahimovs,  Rivosh,  and  Santos~\cite{Nahimovs:2016, Nahimovs:2017} were able to define a set of configurations of marked vertices for which quantum walk search does not have any speed-up over the classical exhaustive search. The reason for this is that for such configurations the initial state of the algorithm (\ref{eq:psi0_gen}) is close to a 1-eigenvector of the search operator $U'$. Therefore, the probability of finding a marked vertex stays close to the initial probability and does not grow over time. 
Instead of analyzing the eigenspectrum of the search operator $U'$ for each configuration of marked vertices the authors of \cite{Nahimovs:2016, Nahimovs:2017} gave the general conditions for a state to be stationary (1-eigenvector of $U'$) in terms of amplitudes of individual vertices and, based on the conditions, they constructed the set of ``bad'' configurations of marked vertices (referred in the papers as {\em exceptional configurations}). 

Another type of exceptional configurations were found  by Ambainis and Rivosh for the two-dimensional lattice~\cite{Ambainis:2008}. In this case, when all vertices on the diagonal are marked, the quantum walk evolves by flipping the signs of the amplitudes of the initial state and the system remains in a uniform probability distribution for all time. Wong and Santos~\cite{Wong:2017} showed that the same  happens for the cycle and any higher-dimensional graph that reduces to the 1D line by using Szegedy's quantum walk model.

Recently, Pr\={u}sis, Vihrovs and Wong~\cite{Prusis:2016} have studied the existence of stationary states on general graphs with multiple marked vertices and found the necessary and sufficient conditions for a set of connected marked vertices to have a stationary state (i.e. to have an assignment of amplitudes which is a 1-eigenvector of the search operator $U'$).

In this paper, we consider a set of connected components (connected subsets) of marked vertices which has a stationary state. We show that the probability of finding a marked vertex is upper bounded by a function of the amplitudes of the stationary state as well as of the properties of the marked components. We give the exact equation of the upper bound function (for both single and multiple marked components) and sketch further research directions. 

The paper is structured as follows. In Section~\ref{sec:stationary} we review the results in the literature by describing in which cases a set of marked vertices forms a stationary state. In Section~\ref{sec:bounds} we study the behaviour of the probability of finding a marked vertex when we have a stationary state.  We draw our conclusions in Section~\ref{sec:conclusions}.


\section{Stationary states}\label{sec:stationary}

To start, let us introduce the notation. Let $G = (V, E)$ be a graph with a connected set of marked vertices $M$. Let $E_M = \{(i,j) \mid i, j\in M \}$ be the set of edges between marked vertices and $E_{\overbar{M}} = E\backslash E_M$ be its complement.
Let $d_i$ be the degree of vertex $i$. Let $d_i^M$ be the number of edges from a vertex $i$ to vertices in $M$ and $d_i^{\overbar{M}}$ be the number of edges from vertex $i$ to vertices in $V\backslash M$.
Trivially, $d_i = d_i^{\overbar{M}}+d_i^M$.
Let $D^{\overbar{M}} = \sum_{i \in M}{d_i^{\overbar{M}}}$ be the total ``outgoing'' degree of a marked set.

A state $\ket{\psi}$ is stationary if $U'\ket{\psi} = \ket{\psi}$. As an example, consider a step of the walk for a cycle of $5$ vertices with $2$ marked vertices shown on Fig. \ref{fig:cycle}.  
\begin{figure}[!htb]
\centering\includegraphics[scale=1]{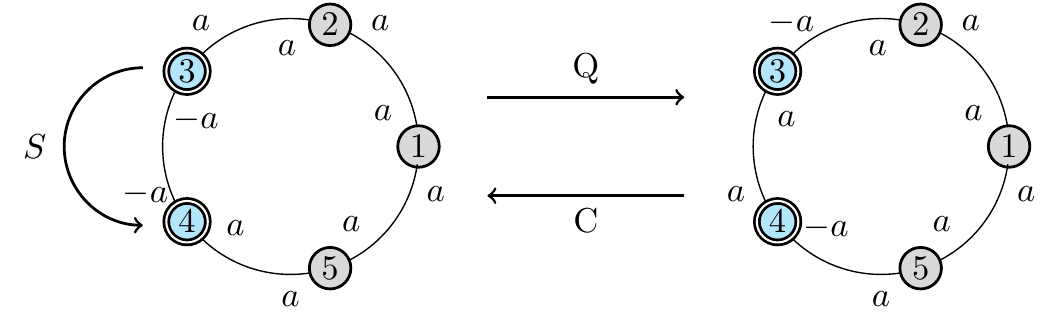}
\caption{Illustration of the application of the evolution operator $U' = SCQ$ to a cycle of 5 vertices with two marked vertices ($M=\{3,4\}$). Labels on edges represent directional amplitudes of a vertex. The state on the left side is a stationary state. The amplitudes of marked vertices pointing to each other are equal to $-a$, all other amplitudes are equal to $a$. In this case, the application of the query operator ($Q$) and the coin operator ($C$) will flip the sign of amplitudes in the marked vertices.}
\label{fig:cycle}
\end{figure}
For simplicity, we will use $\ket{i,j}$ to denote the direction amplitude of vertex $i$ pointing towards vertex $j$. Using this notation the state on the left is written as
\begin{equation*}
\begin{split}
\ket{\psi} &= a(\ket{1,2}+\ket{1,5}) + a(\ket{2,1}+\ket{2,3})+a(\ket{3,2}-\ket{3,4})+a(-\ket{4,3}+\ket{4,5})\\
&+a(\ket{5,4} +\ket{5,1}).
\end{split}
\end{equation*}
and the state on the right as
\begin{equation*}
\begin{split}
\ket{\psi^\prime} &= a(\ket{1,2}+\ket{1,5}) + a(\ket{2,1}+\ket{2,3})-a(\ket{3,2}-\ket{3,4})-a(-\ket{4,3}+\ket{4,5})\\
&+a(\ket{5,4} +\ket{5,1}).
\end{split}
\end{equation*}
When applying the evolution operator $U'=SCQ$ to the state $\ket{\psi}$, the amplitudes are changed from state $\ket{\psi}$ to state $\ket{\psi^\prime}$ and back. The query operator ($Q$) will flip the sign of directional amplitudes of the marked vertices. The coin operator ($C$) will undo the effect of the query operation by flipping the signs again, as amplitudes in the marked vertices add up to zero. And since the directional amplitudes of adjacent vertices pointing to each other are equal, the shift operator ($S$) has no effect on the state. Therefore, $\ket{\psi}$ is not changed by a step of the walk, i.e. it is stationary.

From this example, it is clear why a state with the following properties is stationary. 
\begin{theorem}[\cite{Nahimovs:2017}]\label{teo:stat}
Consider a state $\ket{\psi}$ with the following properties: all amplitudes of the unmarked vertices are equal; the sum of the amplitudes of any marked vertex is $0$; the amplitudes of two adjacent vertices pointing to each other are equal. Then $\ket{\psi}$ is a stationary state of the evolution operator $U^\prime$.
\label{teo:stat}
\end{theorem} 

When a configuration of marked vertices forms a stationary state it is important to know how close is the stationary state to the initial state. 
Depending on that, it may or may not affect the search. Moreover, as we will see next, a configuration of marked vertices may have multiple stationary states.

According to Pr\={u}sis, Vihrovs and Wong~\cite{Prusis:2016}, the existence of a stationary state depends on whether a marked connected component is bipartite or not, that is,
\begin{theorem}[\cite{Prusis:2016}] \label{teo:bip} A bipartite marked connected component has a stationary state if and only if the sums of $d_i^{\overbar{M}}$ for each bipartite set are equal. A non-bipartite marked connected component always has a stationary state.
\label{teo:bip}
\end{theorem}

\noindent
In the example on Fig.~\ref{fig:cycle}, the connected component is bipartite and $d_{3}^{\overbar{M}} = d_{4}^{\overbar{M}} = 1$. Therefore, is forms a stationary state.
 
Now, suppose we have a graph with a marked connected component satisfying the Theorem \ref{teo:bip}. 
It has a stationary state (depicted on Fig.~\ref{fig:stat})
\begin{equation}\label{eq:st}
\ket{\psi_{ST}^a} = \sum_{\substack{i,j\in V\\ j\sim i}}{a \ket{i,j}} + \sum_{\substack{i, j\in M\\ j\sim i}}{(c_{ij} - 1)a\ket{i,j}},
\end{equation}
where $j\sim i$ means there is an edge connecting vertex $j$ to vertex $i$. 
\begin{figure}[!htb]
\centering
\includegraphics[scale=1]{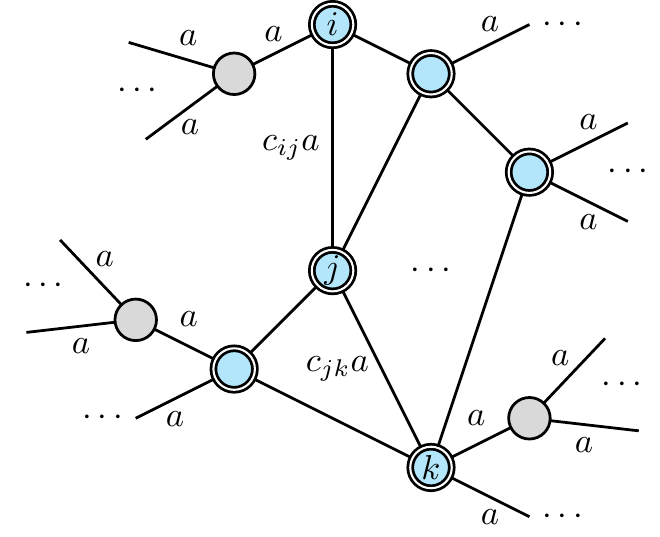}
\caption{Sketch of amplitudes of a stationary state of a marked connected component. The marked vertices are represented by double circles. All amplitudes of unmarked vertices are equal to $a$. For simplicity, we depicted only one of the amplitudes of vertices pointing to each other, since they are equal.}
\label{fig:stat}
\end{figure}
All amplitudes of unmarked vertices (represented by single circles) are equal to $a$, so the coin transformation have no effect on these vertices. The amplitudes of marked vertices (represented by double circles) pointing to unmarked vertices are also equal to $a$. The amplitude of the marked vertex $i$ pointing towards marked vertex $j$ is $c_{ij}a$. Note, that according to Theorem~\ref{teo:stat} $c_{ij} = c_{ji}$, so the shift operator have no effect on the marked component.  Moreover, the sum of the directional amplitudes of a marked vertex must be equal to zero. In this way, the coin operator will flip the sign of the amplitudes, undoing the effect of the query operator. Therefore, the amplitudes $c_{ij}$ should satisfy
\begin{equation}
\sum_{\substack{j\in M\\j\sim i}} c_{ij} = d_i^{\overbar{M}},\quad \forall i\in M.
\end{equation}

A marked connected component may have multiple (infinitely many) stationary states satisfying the properties given by Theorem~\ref{teo:stat}. For example, Fig.~\ref{fig:grid} shows two different stationary states for the same marked connected component in a two-dimensional lattice with $N$ vertices.
\begin{figure}[!htb]
\centering
\subcaptionbox{Stationary state with amplitudes $-3a$ at vertices $6$ and $10$, and $7$ and $11$ pointing to each other.}{
\includegraphics[scale=0.8]{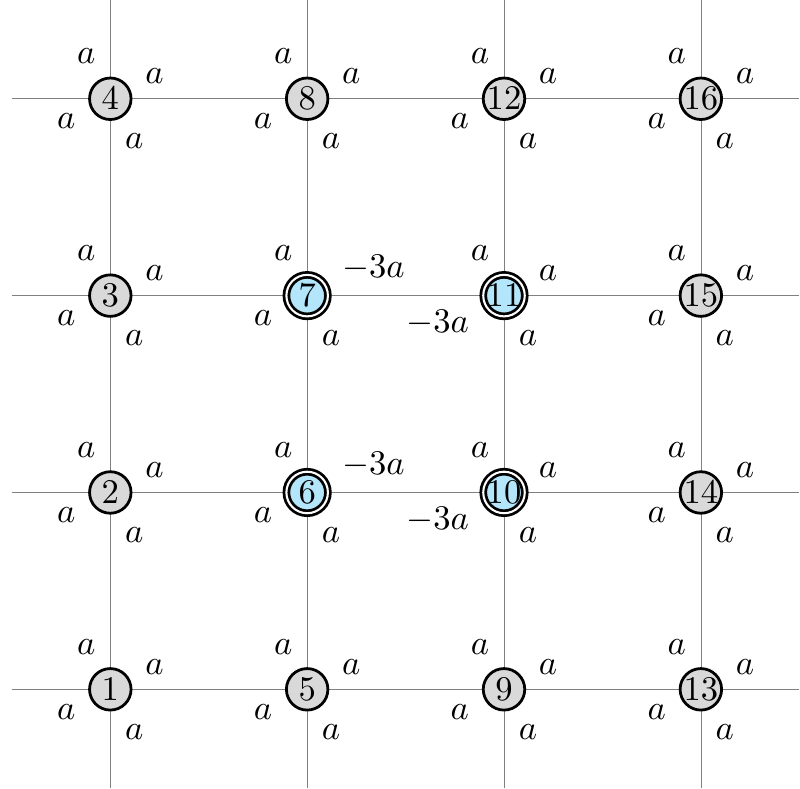}}
\subcaptionbox{Stationary state with amplitudes $-a$ at all marked vertices $\{6,7,10,11\}$ pointing to each other.}{
\includegraphics[scale=0.8]{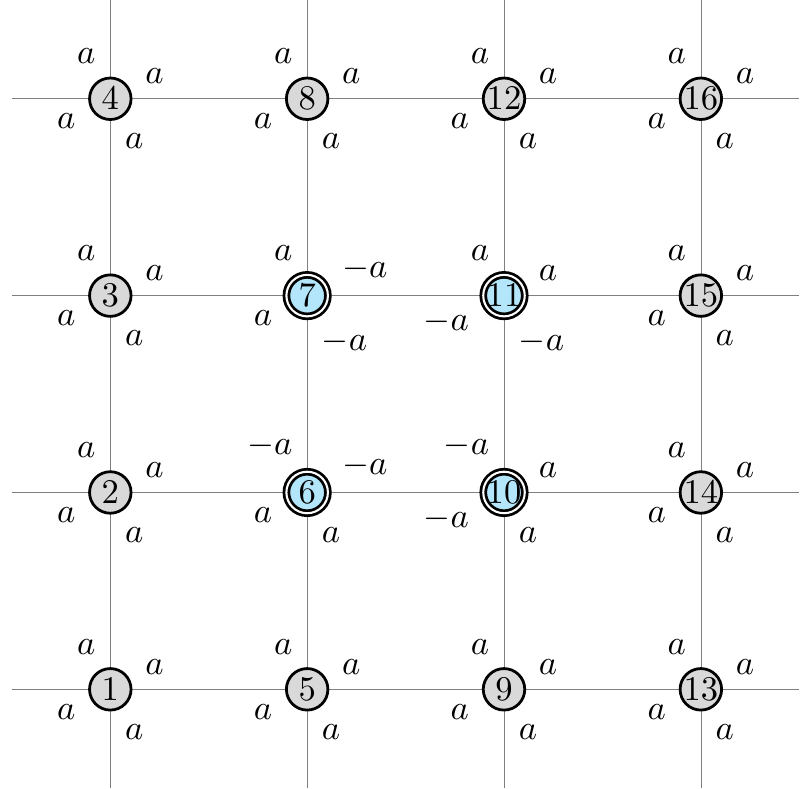}}
\caption{Stationary states in the two-dimensional lattice.}
\label{fig:grid}
\end{figure}

\noindent
It might seem that both states -- $\ket{\psi_1}$ (on the left) and $\ket{\psi_2}$ (on the right) -- have the same overlap with the initial state, however, this is not true.
Although, for the both states the overlap is $4a^2(N-4)$, the value of $a$ itself is different.
As the states are unit vectors, the sum of squares of all amplitudes needs to be $1$. Therefore, value of $a$ for $\ket{\psi_1}$ is $\frac{1}{\sqrt{4(N+8)}}$ and the value of $a$ for $\ket{\psi_2}$ is $\frac{1}{\sqrt{4N}}$.
The probability of finding a marked vertex is $48a^2 = \frac{12}{N+8}$ for $\ket{\psi_1}$ and $16a^2 = \frac{4}{N}$ for $\ket{\psi_2}$.

In the following section, we consider a connected set of marked vertices having a stationary state. We analyze the evolution of a state of the algorithm and we prove an upper bound for the probability of finding a marked vertex.


\section{Bounds on the probability}
\label{sec:bounds}

\subsection{Upper bound on the probability for a connected component of marked vertices}

\begin{theorem}\label{teo:bprob}
Consider a graph $G=(V,E)$ with a connected component of marked vertices $M$. Let $M$ be such that there exists a stationary state.
Then, the probability $p_M$ of finding a marked vertex, for any number of steps $t$, is
\begin{equation}
\label{probability_bound}
p_M \leq \frac{2}{m} \left( \sum_{\substack{i, j\in M\\ j\sim i}}{c_{ij}^2} + 2D^{\overbar{M}} + 2|E_M| \right) .
\end{equation}
\end{theorem}
\begin{proof}

Consider the amplitudes of the stationary state: for all $\ket{i,j}$, where $i \notin M$ or $j \notin M$, the amplitudes are equal to $a$. For $\ket{i,j}$, where $i,j \in M$, let the amplitudes be $c_{ij} \cdot a$. The stationary state, then, can be written as 
$$ \ket{\psi_{ST}^a} = \sum_{\substack{i,j\in V\\ j\sim i}}{a \ket{i,j}} + \sum_{\substack{i, j\in M\\ j\sim i}}{(c_{ij} - 1)a\ket{i,j}} .$$
Thus, we have 
$$ \ket{\psi_0} = \ket{\psi_{ST}^a} - \sum_{\substack{i, j\in M\\ j\sim i}}{(c_{ij} - 1)a\ket{i,j}},$$
for $a=1/\sqrt{2m}$.
We will denote the changing part of the initial state as $\ket{\psi_{NST}}$.

Out task is to upper bound the probability of finding a marked vertex $p_M$, that is, to find a distribution of $\ket{\psi_{NST}}$ over the graph, which maximizes the probability of finding a marked vertex.
Clearly, the probability is maximized when the $\ket{\psi_{NST}}$ is distributed over the marked vertices only. Therefore, our task is to maximize 
\begin{equation}
\label{probatility_bound_pr}
p_M = \sum_{i \in M}{\left[ \sum_{\substack{j \in V\backslash M\\j\sim i}}{(a + \alpha)^2} + \sum_{\substack{j\in M\\j\sim i}}{(c_{ij}a + \alpha_{ij})^2} \right]} 
\end{equation}
subject to
\begin{equation}
\label{probatility_bound_subject}
\sum_{i \in M}{\left[ \sum_{\substack{j \in V\backslash M\\j\sim i}}{\alpha^2} + \sum_{\substack{j\in M\\j\sim i}}{\alpha_{ij}^2} \right]} = ||\ket{\psi_{NST}}||^2, 
\end{equation}
because the evolution is unitary and the norm will not be changed during the evolution. In Eq.~({\ref{probatility_bound_pr}}), $a$ and $c_{ij}a$ come from the stationary part of the state and $\alpha$ and $\alpha_{ij}$ come from the non-stationary part.

Marked vertex $i$ has an ``outgoing'' degree $d_i^{\overbar{M}}$. Therefore, we can rewrite Eqs.~(\ref{probatility_bound_pr})-(\ref{probatility_bound_subject}) as
\begin{equation}
p_M = \sum_{i \in M}{\left[ d_i^{\overbar{M}} (a + \alpha)^2 + \sum_{\substack{j\in M\\j\sim i}}{(c_{ij}a + \alpha_{ij})^2} \right]} 
\end{equation}
subject to
\begin{equation}
\sum_{i \in M}{\left[ d_i^{\overbar{M}} \alpha^2 + \sum_{\substack{j\in M\\j\sim i}}{\alpha_{ij}^2} \right]} = ||\ket{\psi_{NST}}||^2 .
\end{equation}

Let $D^{\overbar{M}} = \sum_{i \in M}{d_i^{\overbar{M}}}$ be the total ``outgoing'' degree of a marked component. Then, we have
\begin{equation}\label{eq:pm}
p_M = D^{\overbar{M}} (a + \alpha)^2 + \sum_{\substack{i, j\in M\\ j\sim i}}{(c_{ij}a + \alpha_{ij})^2} .
\end{equation}
subject to
\begin{equation}
D^{\overbar{M}} \alpha^2 + \sum_{\substack{i, j\in M\\ j\sim i}}{\alpha_{ij}^2} = ||\ket{\psi_{NST}}||^2 .
\end{equation}

Next, we will use the following technical Lemma~\ref{lemma:technical}(proved in Appendix~\ref{ap:proof}).
\begin{lemma}\label{lemma:technical}
Let  $f(X)=\sum_{i=1}^{n}(x_i-a_i)^2$ and $r^2=\sum_{i=1}^{n}x_i^2$. Then 
\[argmax_{X}f=(-d\cdot a_1,\dots,-d\cdot a_n)\mbox{, for }d=\frac{r}{\sqrt{\sum_{i=1}^n a_i^2}}.\]
\end{lemma}

\noindent
From the lemma we have that the probability reaches its maximum for $\alpha = \frac{r}{d} a$ and $\alpha_{ij} = \frac{r}{d} c_{ij}a$, where 
\begin{equation}
\label{probatility_bound_r}
r = ||\ket{\psi_{NST}}|| = \sqrt{\sum_{\substack{i, j\in M\\ j\sim i}}{(c_{ij} - 1)^2 a^2}} 
\end{equation}
and
\begin{equation}
\label{probatility_bound_d}
d = \sqrt{D^{\overbar{M}} a^2 + \sum_{\substack{i, j\in M\\ j\sim i}}{c_{ij}^2 a^2}} .
\end{equation}

Thus, from Eq.~(\ref{eq:pm}), the probability $p_M$ reaches its maximum for  
$$
D^{\overbar{M}} \left(a + \frac{r}{d}a\right)^2 + \sum_{\substack{i, j\in M\\ j\sim i}}{\left(c_{ij}a + \frac{r}{d} c_{ij}a\right)^2} =
$$
$$
\left(1 + \frac{r}{d}\right)^2 \left( D^{\overbar{M}} a^2 + \sum_{\substack{i, j\in M\\ j\sim i}}{c_{ij}^2 a^2} \right) =
$$
$$
\left(1 + \frac{r}{d}\right)^2 d^2 = (d + r)^2 .
$$
Using values for $r$ and $d$ from Eqs.~(\ref{probatility_bound_r}) and (\ref{probatility_bound_d}) we obtain
$$
a^2 \left( \sqrt{D^{\overbar{M}} + \sum_{\substack{i, j\in M\\ j\sim i}}{c_{ij}^2}} + \sqrt{\sum_{\substack{i, j\in M\\ j\sim i}}{(c_{ij} - 1)^2}} \right)^2 .
$$
Opening the brackets (under the second square root) and using that
$$ D^{\overbar{M}} + \sum_{\substack{i, j\in M\\ j\sim i}}{c_{ij}} = 0,$$
we have
$$
a^2 \left( \sqrt{D^{\overbar{M}} + \sum_{\substack{i, j\in M\\ j\sim i}}{c_{ij}^2}} + \sqrt{\sum_{\substack{i, j\in M\\ j\sim i}}{c_{ij}^2} + 2D^{\overbar{M}} + 2|E_M|} \right)^2 .
$$
Therefore,
\begin{equation}
\label{probability_bound}
p_M \leq 4a^2 \left( \sum_{\substack{i, j\in M\\ j\sim i}}{c_{ij}^2} + 2D^{\overbar{M}} + 2|E_M| \right) .
\end{equation}
\end{proof}

Note that the first term in Eq.~(\ref{probability_bound}) depends on the stationary state, while the two others -- on the structure of the graph and the marked component. 

For a given graph and a marked component there are infinitely many stationary states. Each stationary state gives a bound on the probability. For a tight bound one needs to consider the stationary state with the minimal $\sum_{\substack{i, j\in M\\ j\sim i}}{c_{ij}^2}$.
This might be a hard task in the general case, and it is still an open question. 


\subsection{Upper bound on the probability for multiple marked components}

Consider we have a disjoint set of $k$ marked connected components M = $\{M_1\cup M_2\cup \dots\cup M_k\}$.
Let $E_{M_l} = \{(i,j) \in E \mid i,j \in M_l\}$ be the set of edges inside the marked component $M_l$ and let $d_i^{\overbar{M_l}}$ be the number of edges from $i$ to a vertex in $V\backslash M_l$. Then, it easily follows from Theorem~\ref{teo:bprob} that
\begin{corollary} Consider a disjoint set of $k$ connected components of marked vertices $\{M_1, M_2,\dots, M_k\}$ such that there exists a stationary state.
Then, the probability $p_M$ of finding a marked vertex, for any number of steps $t$, is
\begin{equation}
p_M \leq \frac{2}{m}\sum_{l=1}^k\left( \sum_{\substack{i, j\in M_l\\ j\sim i}}{c_{ij}^2} + 2D^{\overbar{M_l}} + 2|E_{M_l}|\right),
\end{equation}
where $D^{\overbar{M_l}} = \sum_{i\in M_l}d_i^{\overbar{M_l}}$.
\end{corollary}

For example, if we consider a $d$-regular graph with a set of marked vertices which consist of $k$ pairs of adjacent marked vertices (i.e. $|M_1| = |M_2| = \dots = |M_k| = 2$). Then, the probability of finding a  marked vertex, for any number of steps $t$, is $O\left(\frac{kd^2}{m}\right)$, where $m$ is the number of edges of the graph. Note that $D^{\overbar{M_l}} = 2(d-1)$ and $|E_{M_l}|=1$ for all $l=1,\dots ,k$.






\section{Conclusions}\label{sec:conclusions}

Due to the interference phenomena, quantum walks behave differently from classical random walks. On the one hand, it can achieve a quadratic speed-up when searching for one marked vertex in a graph. On the other hand, additional marked vertices can make the search harder. We have seen that a placement of marked vertices on a graph can form a stationary state. However, having a stationary state does not automatically mean that the quantum search will not be able to find a marked vertex faster than classically. That is why we need to understand how the probability of finding a marked vertex behaves during the evolution. We proved that the probability is upper bounded by a function on the amplitudes of the stationary state and on the structure of the marked components. 

As we have seen, there are infinitely many stationary states for a given set of marked connected components.
It is still an open problem to find which stationary state gives the minimum probability to find a marked vertex. In this way, we can obtain a tighter bound on the probability.

Another interesting question, is whether we can find applications for the exceptional configurations. One idea is to solve the problem of bipartite matching. Given a bipartite graph, the goal is to determine whether a perfect matching exists. Our initial idea is to embed the graph into the two-dimensional lattice and make all its vertices marked. We will need to make some restrictions in the graph for that. Then, we claim that a perfect matching will exist if the marked vertices forms a stationary state. We plan to investigate this problem in the near future.


\subparagraph*{Acknowledgements.}
The authors thank A. Ambainis, A. Rivosh, K. Pr\={u}sis and J. Vihrovs for useful discussions and comments.

This work was supported by the RAQUEL (Grant Agreement No.~323970) project, the Latvian State Research Programme NeXIT
project No.~1, the ERC Advanced Grant MQC and ERDF project number 1.1.1.2/VIAA/1/16/002.


\bibliography{refs}

\begin{thebibliography}{10}

\bibitem{Aharonov:1993}
Y.~Aharonov, L.~Davidovich, and N.~Zagury.
\newblock Quantum random walks.
\newblock {\em Physical Review A}, 48(2):1687--1690, 1993.

\bibitem{Ambainis:2008}
A.~Ambainis and A.~Rivosh.
\newblock Quantum walks with multiple or moving marked locations.
\newblock In {\em Proceedings of SOFSEM}, pages 485--496, 2008.

\bibitem{Grover:1996}
L.~K. Grover.
\newblock A fast quantum mechanical algorithm for database search.
\newblock In {\em Proceedings of the 28th ACM Symposium on the Theory of
  Computing}, pages 212--219, 1996.

\bibitem{Hoyer:2017}
Peter Hoyer and Mojtaba Komeili.
\newblock Efficient quantum walk on the grid with multiple marked elements.
\newblock In Heribert Vollmer and Brigitte Valle, editors, {\em 34th Symposium
  on Theoretical Aspects of Computer Science (STACS 2017)}, volume~66 of {\em
  Leibniz International Proceedings in Informatics (LIPIcs)}, pages
  42:1--42:14, Dagstuhl, Germany, 2017. Schloss Dagstuhl--Leibniz-Zentrum fuer
  Informatik.
\newblock \href {http://dx.doi.org/10.4230/LIPIcs.STACS.2017.42}
  {\path{doi:10.4230/LIPIcs.STACS.2017.42}}.

\bibitem{Krovi:2016}
Hari Krovi, Fr{\'e}d{\'e}ric Magniez, Maris Ozols, and J{\'e}r{\'e}mie Roland.
\newblock Quantum walks can find a marked element on any graph.
\newblock {\em Algorithmica}, 74(2):851--907, 2016.
\newblock \href {http://dx.doi.org/10.1007/s00453-015-9979-8}
  {\path{doi:10.1007/s00453-015-9979-8}}.

\bibitem{Magniez:2012}
Fr{\'e}d{\'e}ric Magniez, Ashwin Nayak, Peter~C. Richter, and Miklos Santha.
\newblock On the hitting times of quantum versus random walks.
\newblock {\em Algorithmica}, 63(1):91--116, 2012.
\newblock \href {http://dx.doi.org/10.1007/s00453-011-9521-6}
  {\path{doi:10.1007/s00453-011-9521-6}}.

\bibitem{Nahimovs:2016}
Nikolajs Nahimovs and Alexander Rivosh.
\newblock Exceptional configurations of quantum walks with {G}rover's coin.
\newblock In {\em Proceedings of the 10th International Doctoral Workshop on
  Mathematical and Engineering Methods in Computer Science}, MEMICS 2015, pages
  79--92, Tel{\v{c}}, Czech Republic, 2016. Springer.
\newblock \href {http://dx.doi.org/10.1007/978-3-319-29817-7_8}
  {\path{doi:10.1007/978-3-319-29817-7_8}}.

\bibitem{Nahimovs:2017}
Nikolajs Nahimovs and Raqueline A.~M. Santos.
\newblock Adjacent vertices can be hard to find by quantum walks.
\newblock In {\em SOFSEM 2017: Theory and Practice of Computer Science: 43rd
  International Conference on Current Trends in Theory and Practice of Computer
  Science, Limerick, Ireland, January 16-20, 2017, Proceedings}, pages
  256--267, Cham, 2017. Springer International Publishing.
\newblock \href {http://dx.doi.org/10.1007/978-3-319-51963-0_20}
  {\path{doi:10.1007/978-3-319-51963-0_20}}.

\bibitem{Prusis:2016}
Kri\v{s}j\={a}nis Pr\={u}sis, Jevg\={e}nijs Vihrovs, and Thomas~G. Wong.
\newblock Stationary states in quantum walk search.
\newblock {\em Phys. Rev. A}, 94:032334, Sep 2016.
\newblock \href {http://dx.doi.org/10.1103/PhysRevA.94.032334}
  {\path{doi:10.1103/PhysRevA.94.032334}}.

\bibitem{Szegedy:2004}
M.~Szegedy.
\newblock Quantum speed-up of \uppercase{M}arkov chain based algorithms.
\newblock In {\em Proceedings of the 45th Symposium on Foundations of Computer
  Science}, pages 32--41, 2004.

\bibitem{Wong:2017}
Thomas~G. Wong and Raqueline A.~M. Santos.
\newblock Exceptional quantum walk search on the cycle.
\newblock {\em Quantum Information Processing}, 16(6):154, 2017.
\newblock \href {http://dx.doi.org/10.1007/s11128-017-1606-y}
  {\path{doi:10.1007/s11128-017-1606-y}}.

\end{thebibliography}


\appendix

\section{Proof of Technical Lemma}\label{ap:proof}

\begin{replemma}{lemma:technical}
Let  $f(X)=\sum_{i=1}^{n}(x_i-a_i)^2$ and $r^2=\sum_{i=1}^{n}x_i^2$. Then 
\[argmax_{X}f=(-d\cdot a_1,\dots,-d\cdot a_n)\mbox{, for }d=\frac{r}{\sqrt{\sum_{i=1}^n a_i^2}}.\]
\end{replemma}
\begin{proof}

Let  $f(X)=\sum_{i=1}^{n}(x_i-a_i)^2$. We want to find the maximal value of $f(X)$, such that $r^2=\sum_{i=1}^{n}x_i^2$.
Let us find $argmax_{X}f$.
Observe that $r^2=\sum_{i=1}^{n}x_i^2$ is the equation of a $(n-1)$-sphere, denote it $S_1$. Note that the center of $S_1$ is $(0,\dots,0)$. $f(X)$ is the radius of a $(n-1)$-sphere, denote it $S_2$. We should find the maximal radius of the sphere such that $S_1$ and $S_2$ still have common points.

Let point $O$ be the center of $S_1$, point $A$ be the center of $S_2$ and $B$ be the intersection point of the spheres. Then, $|OB|=r$, $|AB|=\sqrt{f}$. Using the Triangle Inequality we can say that $|AB|\leq |OB|+|AO|$. It means that $|AB|$ will achieve its maximum when $|AB|=|OB|+|AO|$, therefore $B$ belongs to the line $OA$.
\begin{center}
\includegraphics[scale=0.5]{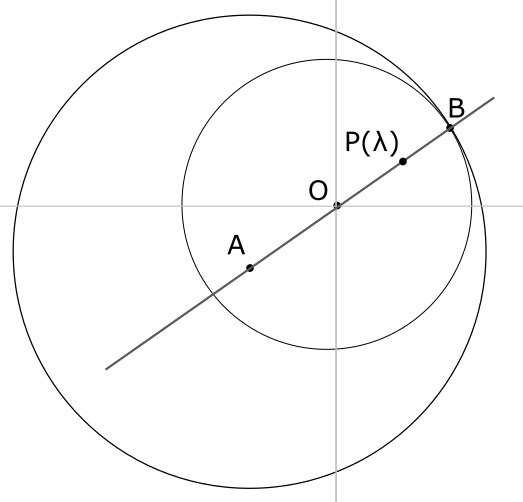}
\end{center}

Let us consider coordinates of any point on the line $OA$. This is: $P(\lambda)=(\lambda a_1,\dots,\lambda a_n)$ for some real $\lambda$. Note that $P(1)=A$. Let us compute $\lambda_0$, such that $P(\lambda_0)=B$.
The length of segment $OP(\lambda)$ is
$$
|OP(\lambda)|=\sqrt{\sum_{i=1}^n(\lambda a_i)^2}=|\lambda|  \sqrt{\sum_{i=1}^n a_i^2}.
$$
Recall, that $|OP(\lambda_0)|=|OB|=r$, therefore
$$
r=|\lambda_0| \sqrt{\sum_{i=1}^n a_i^2},\mbox{ and } 
|\lambda_0|=\frac{r}{\sqrt{\sum_{i=1}^n a_i^2}}.
$$

Note that $\lambda_0<0$ because $|P(1)P(|\lambda_0|)|<|P(1)P(-|\lambda_0|)|$.
Therefore, $B=(-d\cdot a_1,\dots,-d\cdot a_n)$ for 
$$
d=\frac{r}{\sqrt{\sum_{i=1}^n a_i^2}}.
$$
In other words,
 \[argmax_{X}f=(-d\cdot a_1,\dots,-d\cdot a_n)\mbox{, for }d=\frac{r}{\sqrt{\sum_{i=1}^n a_i^2}}.\]
 
\end{proof}

\end{document}